\documentclass[a4paper,UKenglish,cleveref, autoref, thm-restate]{lipics-v2021}

\usepackage[ruled,vlined,linesnumbered,commentsnumbered]{algorithm2e}
\usepackage{xcolor}
\usepackage{hyperref}

\nolinenumbers

\title             {Tight analysis of the primal-dual method for edge-covering pliable set families}
\titlerunning{Tight analysis of the primal-dual method for edge-covering pliable set families}

\author{Zeev Nutov}{The Open University of Israel}{nutov@openu.ac.il}
{https://orcid.org/0000-0002-6629-3243}{}
\authorrunning{Zeev Nutov}
\Copyright{Zeev Nutov}


\begin{document}

\maketitle


\newcommand {\ignore} [1] {}

\newcommand{\sem}    {\setminus}
\newcommand{\subs}   {\subseteq}
\newcommand{\empt}   {\emptyset}

\newcommand{\f}   {\frac}

\def\al   {\alpha}
\def\be {\beta}
\def\ga {\gamma}
\def\de   {\delta}
\def\eps {\epsilon}
\def\la    {\lambda}
\def\th    {\theta}

\def\CC  {{\cal C}}
\def\FF   {{\cal F}}
\def\LL   {{\cal L}}
\def\RR  {{\cal R}}
\def\SS   {{\cal S}}
\def\TT   {{\cal T}}

\def\sfec    {{\sc Set Family Edge Cover}}
\def\nmcc  {{\sc Near Min-Cuts Cover}}
\def\ckecs  {{\sc Cap-$k$-ECSS}}               
\def\fgc       {{\sc FGC}}                                 

\keywords{primal dual method, pliable set family, approximation algorithms} 

\begin{abstract}
A classic result of Williamson, Goemans, Mihail, and Vazirani [STOC 1993: 708--717]  
states that the problem of covering an uncrossable set family by a min-cost edge set 
admits approximation ratio~$2$, by a primal-dual algorithm with a reverse delete phase.
Bansal, Cheriyan, Grout, and Ibrahimpur [ICALP 2023: 15:1–15:19] showed that 
this algorithm achieves approximation ratio $16$ for a larger class of so called $\gamma$-pliable set families,
that have much weaker uncrossing properties. 
The approximation ratio $16$ was improved to $10$ in \cite{N-prel}.
Recently, Bansal \cite{B} stated 
approximation ratio $8$ for $\ga$-pliable families and an improved 
approximation ratio $5$ for an important particular case of the family of cuts of size $<k$ of a graph $H$,
but his proof has an error.
We will improve the approximation ratio to $7$ for the former case 
and give a simple proof of approximation ratio $6$ for the latter case.
Furthermore, if $H$ is $\la$-edge-connected  
then we will show a slightly better approximation ratio $6-\f{1}{\be+1}$, 
where $\be=\left\lfloor \f{k-1}{\lceil(\la+1)/2\rceil}\right\rfloor$.
Our analysis is supplemented by examples showing that 
these approximation ratios are tight for the primal-dual algorithm.
\end{abstract}

\section{Introduction} \label{s:intro}

For an edge set or a graph $J$ on node set $V$ and disjoint node subsets $S,T \subs V$ let $\de_J(S,T)$ denote the set 
of edges in $J$ between $S$ and $T$, and let $d_J(S,T)=|\delta_J(S,T)|$ be their number;
we let $\de_J(S)=\de_J(S,V\sem S)$ and $d_J(S)=d_J(S,V\sem S)$.
An edge set $J$ {\bf covers} $S$ if $d_J(S) \geq 1$.
The following generic meta-problem captures dozens of specific  network design problems,
among them {\sc Steiner Forest}, {\sc $k$-Constrained Forest}, {\sc Point-to-Point Connection}, 
{\sc Steiner Network Augmentation}, and many more.  

\begin{center}
\fbox{\begin{minipage}{0.98\textwidth} \noindent
\underline{\sc Set Family Edge Cover} \\ 
{\em Input:} \ \ A graph $G=(V,E)$ with edge costs $\{c_e:e \in E\}$, a set family $\FF$ on $V$. \\ 
{\em Output:} A min-cost edge set $J \subs E$ such that $d_J(S) \geq 1$ for all $S \in \FF$.
\end{minipage}} \end{center}

In this problem the family $\FF$ may not be given explicitly, but we will require that some queries related to $\FF$ can be 
answered in time polynomial in $n=|V|$. 
An inclusion-minimal set in $\FF$ is called an $\FF$-{\bf core}, or just a {\bf core}, if $\FF$ is clear from the context. 
Following previous work, we will require that for any edge set $J$, 
the cores of the {\bf residual family  $\FF^J=\{S \in \FF:d_J(S)=0\}$} of $\FF$ 
(the family of sets in $\FF$ that are uncovered by $J$) can be computed in time polynomial in $n=|V|$.

Agrawal, Klein and Ravi \cite{AKR} designed and analyzed a primal-dual algorithm
for the {\sc Steiner Forest} problem, and showed that it achieves approximation ratio $2$. 
A classic result of Goemans and Williamson \cite{GW} from the early 90's shows by an elegant proof 
that the same algorithm applies for proper set families, where $\FF$ is {\bf proper} if 
it is {\bf symmetric} ($A \in \FF$ implies $V \sem A \in \FF$) and has the {\bf disjointness property} 
(if $A,B$ are disjoint and $A \cup B \in \FF$ then $A \in  \FF$ or $B \in \FF$).
Slightly later, Williamson, Goemans, Mihail, and Vazirani \cite{WGMV} (henceforth WGMV) 
further extended this result to the more general class of {\bf uncrossable families} 
($A \cap B,A \cup B \in \FF$ or $A \sem B,B \sem A \in \FF$ whenever $A,B \in \FF$),
by adding to the algorithm a novel reverse-delete phase.
They posed an open question of extending this algorithm to 
a larger class of set families and combinatorial optimization problems.
However, for 30 years, the class of uncrossable set families remained the most general generic class 
of set families for which the WGMV algorithm achieves a constant approximation ratio.

Bansal, Cheriyan, Grout, and Ibrahimpur \cite{BCGI} (henceforth BCGI) 
analyzed the performance of the WGMV algorithm \cite{WGMV}
for the following generic class of set families that arise in variants of capacitated network design problems.  

\begin{definition}
Two sets $A,B$ {\bf cross} if all the sets $A \cap B,V \sem (A \cup B),A \sem B,B \sem A$ are non-empty.
A set family $\FF$ is {\bf pliable} if $\empt,V \notin \FF$ and for any $A,B \in \FF$ 
at least two of the sets $A \cap B,A \cup B, A\sem B,B \sem A$ belong to $\FF$.  
We say that $\FF$ is {\bf $\ga$-pliable} if it has the~following 
{\bf Property $(\gamma)$:} 
For any edge set $I$ and sets $S_1 \subset S_2$ in the residual family $\FF^I$,  
if $\FF^I$-core (an inclusion minimal set in  $\FF^I$) $C$ crosses each of $S_1,S_2$,
then the set $D=S_2 \sem (S_1 \cup C)$ is either empty or belongs to $\FF^I$.
\end{definition}

BCGI showed that the WGMV algorithm 
achieves approximation ratio $16$ for $\gamma$-pliable families, and that Property $(\gamma)$ is essential -- 
without it the cost of a solution found by the WGMV algorithm can be
$\Omega(\sqrt{n})$ times the cost of an optimal solution.
Another generalization of uncrossable families is considered in \cite{N-prd}.
A set family $\FF$ is {\bf semi-uncrossable} if for any $A,B \in \FF$ we have that
$A \cap B \in \FF$ and one of $A \cup B,A \sem B,B \sem A$ is in $\FF$, or $A \sem B,B \sem A \in \FF$. 
One can verify that semi-uncrossable families are sandwiched between uncrossable and $\ga$-pliable families. 
The WGMV algorithm achieves the same approximation ratio $2$ for semi-uncrossable families,
and \cite{N-prd} shows that many problems can be modeled by semi-uncrossable families that are not uncrossable. 

The approximation ratio $16$ of BCGI \cite{BCGI} for $\ga$-pliable families 
was improved to $10$ in \cite{N-prel}; in fact, the analysis in \cite{N-prel} 
immediately implies approximation ratio $9$, see Lemma~\ref{l:9}.
Recently Bansal \cite{B} stated an approximation ratio of $8$. 
Here we improve the approximation ratio to $7$, and show that this bound is 
asymptotically  tight for the WGMV algorithm.

\begin{theorem} \label{t:1}
The {\sc Set Family Edge Cover} problem with a $\gamma$-pliable set family $\FF$ 
admits approximation ratio $7$.
\end{theorem}

A set family $\FF$ is {\bf sparse} if for any edge set $J$, 
every set $S \in \FF^J$ crosses at most one $\FF^J$-core.
A particular important case of  $\ga$-pliable families arise from  the {\sc Small Cuts Cover} problem, 
in which we seek to cover by a min-cost edge set the set family 
$\FF=\{\empt \neq S \subset V: d_H(S) < k\}$ of cuts of size/capacity $<k$ of a graph $H$. 
Bansal \cite{B} made an important observation that this family is sparse.
To see this, note that if $S \in \FF$ crosses an $\FF$-core $C$ then
$C \cap S \notin \FF$ and $C \sem S \notin \FF$ by the minimality of $C$, 
hence $d_H(C \cap S) \ge k$ and $d_H(C \sem S) \ge k$. Thus we have
$2d_H(C \cap S,C \sem S)=d_H(C \cap S)+d_H(C \sem S)-d_H(C) \ge k+1$.
One can see that the cores are pairwise disjoint, hence if $S$ crosses two cores $C_1,C_2$ then 
$d_H(S) \ge d_H(C_1 \cap S,C_1 \sem S)+d_H(C_2 \cap S,C_2 \sem S) \ge 2 \lceil (k+1)/2 \rceil$, contradicting that $d_G(S)<k$. 
Since $\FF^J$ is the family of cuts of size $<k$ of the graph $H \cup J$, we get that this $\FF$ is sparse.

Bansal \cite{B} also stated an approximation ratio of $5$ for $\ga$-pliable sparse families,
but his proof has an error \cite{Bpm}.
We note that in one of earlier versions v2 of his arXiv draft \cite{B}, along with the 
$5$-approximation for {\sc Small Cuts Cover} he also states a $6$-approximation for $\ga$-pliable sparse families. 
The proof provided relies on several complex reductions and decompositions and many phases of token distribution.
We give a relatively simple proof of a $6$-approximation in this case, 
that relies on a clear combinatorial statement (Lemma~\ref{l:pliable}), 
and also give an example that this bound is asymptotically tight for the WGMV algorithm.

We will also investigate the dependence of the approximation ratio on the ``inverse'' parameter -- 
the maximum number of pairwise disjoint sets in $\FF^J$ that a single core can cross. 
We say that a set family $\FF$ is {\bf $\be$-crossing} for an integer $\be \ge 1$ if for any
edge set $J$, an $\FF^J$-core crosses at most $\be$ pairwise disjoint sets in $\FF^J$.

\begin{theorem} \label{t:2}
The {\sc Set Family Edge Cover} problem with a $\gamma$-pliable sparse set family $\FF$ 
admits approximation ratio $6$.
If in addition $\FF$ is $\beta$-crossing then the approximation ratio is $6-\f{1}{\be+1}$.
\end{theorem}

The family of cuts of size/capacity $<k$ of a $\la$-edge-connected graph $H$
is $\be$-crossing for $\be=\left \lfloor \f{k-1}{\lceil(\la+1)/2\rceil}\right\rfloor$.
To see this, note that if $S \in \FF$ crosses an $\FF$-core $C$ then 
$d_H(S \cap C) \ge k$ since $S \cap C \notin \FF$, 
$d_H(S \sem C) \ge \la$ since $H$ is $\la$-edge-connected, and 
$d_H(S) \le k-1$ since $S \in \FF$.
Thus we have 
$2d_H(S \cap C,S \sem C) = d_H(S \cap C)+d_H(S \sem C)-d_H(S) \ge \la+1$.
If each of $p$ disjoint sets $S_1, \ldots,S_p$ in $\FF$ crosses $C$, then 
each $S_i$ contributes to $\de_H(C)$ the set $F_i=\de_H(S_i \cap C,S_i \sem C)$ 
of at least $|F_i| \ge \lceil (\la+1)/2 \rceil$ edges.
The edge sets $F_i$ are pairwise disjoint, thus
$k-1 \ge d_H(C) \ge p\cdot \lceil (\la+1)/2 \rceil$.
Combined with Theorem~\ref{t:2} we get:

\begin{corollary} \label{c:be}
The {\sc Small Cuts Cover} problem with $\la$-edge-connected graph $H$ 
admits approximation ratio $6-\f{1}{\be+1}$, 
where $\be=\left\lfloor \f{k-1}{\lceil(\la+1)/2\rceil}\right\rfloor$.
\end{corollary}

The proofs of Theorems \ref{t:1} and \ref{t:2} rely on a new structural property 
of inclusion-minimal solutions that was not known prior to this paper, see Lemma~\ref{l:pliable}.

For additional applications of $\ga$-pliable families for the so called {\sc Flexible Graph Connectivity} problems see, 
for example, \cite{AHM,BCHI,CJ,BCGI,N-prel,B,BCKS}.
In particular, the second part of Theorem~\ref{t:2} can be used to improve approximation ratios for this problem, 
c.f. \cite{B}. 

The rest of this paper is organized as follows. In the next section we will describe the WGMV 
primal-dual algorithm for pliable set families and show that its approximation ratio is determined by a certain 
combinatorial problem. 
Theorems \ref{t:1} and \ref{t:2} are proved in 
Sections \ref{s:t} and \ref{s:sparse}, respectively.

\section{The WGMV algorithm and pliable families} \label{s:prd}

We start by describing the WGMV algorithm for an arbitrary set family $\FF$. 
Recall that an inclusion-minimal set in $\FF$ is called an $\FF$-{\bf core}, 
or just a {\bf core}, if $\FF$ is clear from the context; let $\CC_\FF$ denote the family of $\FF$-cores. 
Consider the following LP-relaxation {\bf (P)} for {\sc Set Family Edge Cover} and its dual program {\bf (D)}:
\[ \displaystyle
\begin{array} {lllllll} 
& \hphantom{\bf (P)} & \min              & \ \displaystyle \sum_{e \in E} c_e x_e                                                                 & 
    \hphantom{\bf (P)} & \max            & \ \displaystyle \sum_{S \in \FF} y_S  \\
& \mbox{\bf (P)}         & \ \mbox{s.t.} & \displaystyle \sum_{e \in \delta(S)} x_e  \geq 1   \ \ \ \forall S \in \FF \ \ \ \ \ \ \  &
    \mbox{\bf (D)}        & \ \mbox{s.t.} & \displaystyle \sum_{\delta(S) \ni e} y_S \leq c_e \ \ \ \forall e \in E \\
& \hphantom{\bf (P)} &                      & \ \ x_e \geq 0 \ \ \ \ \ \ \ \ \ \forall e \in E                                                                  &
    \hphantom{\bf (P)} &                      & \ \ y_S \geq 0 \ \ \ \ \ \ \ \ \ \ \forall S \in \FF 
\end{array}
\]

Given a solution $y$ to {\bf (D)}, an edge $e \in E$ is {\bf tight} if the inequality of $e$ in {\bf (D)} holds with equality.
The algorithm has two phases.

{\bf Phase~1} starts with $J=\emptyset$ an applies a sequence of iterations.
At the beginning of an iteration, we compute the family $\CC=\CC_{\FF^J}$ of $\FF^J$-cores.
Then we raise the dual variables corresponding to the $\FF^J$-cores
uniformly (possibly by zero), until some edge $e \in E \sem J$ becomes tight, and add $e$ to $J$.
Phase~1 terminates when $\CC_{\FF^J}=\empt$, namely when $J$ covers $\FF$.

{\bf Phase~2} is a ``reverse delete'' phase, in which we process edges in the reverse order
that they were added, and delete an edge $e$ from $J$ if $J \sem \{e_i\}$ still covers $\FF$.
At the end of the algorithm, $J$ is output.

The produced dual solution is feasible, hence 
$\sum_{S \in \FF} y_S \leq {\sf opt}$, by the Weak Duality Theorem.
To prove an approximation ratio of $\rho$, it is sufficient  
to prove that at the end of the algorithm the following holds for the returned solution $J$
and the dual solution $y$:
$$
\sum_{e \in J} c(e) \leq \rho \sum_{S \in \FF} y_S \ .
$$
As any edge in the solution $J$ returned by the algorithm is tight, this is equivalent to
$$
\displaystyle \sum_{e \in J} \sum_{\delta_J(S) \ni e} y_S \leq \rho \sum_{S \in \FF} y_S \ .
$$
By changing the order of summation we get:
$$
\sum_{S \in \FF} d_J(S) y_S \leq \rho \sum_{S \in \FF} y_S \ .
$$
It is sufficient to prove that at any iteration the increase at the left hand side is at most
the increase in the right hand side. 
Let us fix some iteration, and let $\CC$ be the family of cores at the beginning of this iteration. 
The increase in the left hand side is 
$\varepsilon \cdot \sum_{C \in \CC} d_J(C)$, 
where $\varepsilon$ is the amount by which the dual variables were raised in the iteration, 
while the increase in the right hand side is 
$\varepsilon \cdot \rho |\CC|$. Consequently, it is sufficient to prove that
$$
\sum_{C \in \CC} d_J(C) \leq \rho |\CC| \ .
$$

Let us use the following notation.
\begin{itemize}
\item
$J_0$ is the set of edges picked at Phase 1 before the current iteration.
\item
$I'=J \sem J_0$ is the set of edges picked after $J_0$ and survived the reverse-delete phase.
\item
$I=\bigcup_{C \in \CC} \de_{I'}(C)$ is the set of edges in $I'$ that cover some $C \in \CC$.
\end{itemize}

\begin{lemma}
Let $\FF'$ be the residual family of $\FF$ w.r.t. $J_0 \cup (I' \sem I)$. Then:
\begin{enumerate}[(i)]
\item
$I$ is an inclusion-minimal cover of $\FF'$. 
\item
$\CC$ is the family of $\FF'$-cores, namely, $\CC=\CC(\FF')$.
\end{enumerate}
\end{lemma}
\begin{proof}
Let $\FF_0=\FF^{J_0}$ be the residual family of $\FF$ w.r.t. $J_0$,
and note that $\FF'$ is the residual family of $\FF_0$ w.r.t. $I' \sem I$.

We prove (i). 
Since the edges were deleted in reverse order,
the edges in $I'$ were considered for deletion when all edges in $J_0$ were still present.
Thus $I'$ is an inclusion-minimal cover of $\FF_0$.
This implies that $I$ is an inclusion-minimal cover of the residual family of $\FF_0$ w.r.t. $I' \sem I$
(this is so for any $I \subs I'$), which is $\FF'$.

We prove (ii).
By the definition, $\CC$ is the family of $\FF_0$-cores. 
No $C \in \CC$ is covered by $I' \sem I$, hence $\CC \subs \CC(\FF')$. 
This also implies that $\FF'$ has no other core $C' \notin \CC(\FF') \sem \CC$, as
otherwise $C' \in \FF_0$ and thus properly contains some $C \in \CC$, which is a contradiction.
\end{proof}

Observing that $d_J(C)=d_I(C)$ for all $C \in \CC$ (since no $C \in \CC$ is covered by $J_0 \cup (I' \sem I)$),
we have the following.

\begin{lemma} \label{l:2C'}
The WGMV primal-dual algorithm achieves approximation ratio $\rho$ 
if for any residual family $\FF'$ of $\FF$ the following holds.
If $\CC$ is the family of $\FF'$-cores and $I$ is an inclusion minimal cover of $\FF'$ 
such that every edge in $I$ covers some $C \in \CC$ then
\begin{equation} \label{e:2C}
\sum_{C \in \CC} d_I(C) \leq \rho |\CC| \ .
\end{equation}
\end{lemma}
 
One can see that if an edge $e$ covers one of the sets $A \cap B, A \cup B, A\sem B,B \sem A$ 
then it also covers one of $A,B$. This implies the following.

\begin{lemma} \label{l:res}
If $\FF$ is pliable or $\ga$-pliable, then so is any residual family $\FF'$ of $\FF$.
\end{lemma}

Due to Lemmas \ref{l:2C'} and \ref{l:res},
to prove that the WGMV algorithm achieves approximation ratio $7$ for a $\ga$-pliable family $\FF$, 
it is sufficient to prove the following purely combinatorial statement. 

\begin{lemma} \label{l:2C}
Let $I$ be an inclusion minimal cover of a $\gamma$-pliable set family $\FF$ such that 
every edge in $I$ covers some $C \in \CC$. Then 
\begin{equation} \label{e:7}
\sum_{C \in \CC} d_I(C) \leq 7 |\CC| \ .
\end{equation}
\end{lemma}

For the proof of Lemma~\ref{l:2C} we need the following simple lemma.

\begin{lemma} \label{l:CS}
Let $\FF$ be a pliable set family and let $S \in \FF$ and $C \in \CC_\FF$ such that  $C \cap S \ne \empt$.
Then either $C \subs S$ or $C,S$ cross and the following holds: 
$S \sem C,S \cup C \in \FF$ and $C \cap S, C \sem S \notin \FF$.
Consequently, the members of $\CC_\FF$ are pairwise disjoint.  
\end{lemma}
\begin{proof}
Suppose that $C$ is not a subset of $S$.
Then $C \sem S \ne \empt$. Also $S \sem C \ne \empt$, since $S$ cannot be a subset of $C$.
By the minimality of $C$ we must have $C \cap S, C \sem S \notin \FF$, 
thus since $\FF$ is pliable $S \sem C,S \cup C \in \FF$.
In particular, $S,C$ cross. 
\end{proof}

A set family $\LL$ is {\bf laminar} if any two sets in $\LL$ are disjoint or one of them contains the other.
Let $I$ be an inclusion minimal edge cover of a set family $\FF$. 
We say that a set $S_e \in \FF$ is a {\bf witness set} for an edge $e \in I$ if 
$e$ is the unique edge in $I$ that covers $S_e$, namely, if $\delta_I(S_e)=\{e\}$.
We say that $\LL \subs \FF$ is a {\bf witness family} for $I$ if 
$|\LL|=|I|$ and for every $e \in I$ there is a witness set $S_e \in \LL$.
By the minimality of $I$, there exists a witness family $\LL \subs \FF$.
The following was proved in BCGI \cite{BCGI}.

\begin{lemma}[BCGI \cite{BCGI}] \label{l:witness}
Let $I$ be an inclusion minimal cover of a pliable set family $\FF$. 
Then there exists a witness family $\LL \subs \FF$ for $I$ that is laminar.
\end{lemma}

Augment $\LL$ by the set $V$. 
A set $S \in \LL$ {\bf owns} a set $C$ if $S$ is the inclusion-minimal set in $\LL$ that contains $C$.
We assign colors to sets in $\LL$ as follows: a set is {\bf black} if it owns some core and is {\bf white} otherwise.

\begin{definition} \label{d:chain}
A sequence $\SS=(S_1, \ldots, S_\ell)$ of sets in $\LL \sem \{V\}$ is a {\bf white chain} 
if each of $S_1, \ldots, S_\ell$ is white and has exactly one child, where $S_{i-1}$ is the child of $S_i$, $i=2, \ldots,\ell$. 
We denote the child of $S_1$ by $S_0$. 
The edge set of $\SS$ is $I_{\SS}=\{a_0b_1, \ldots,a_\ell b_{\ell+1}\}$, 
where $a_ib_{i+1}$ is the unique edge in $I$ that covers $S_i$ and $a_i,b_i \in S_i$;
see Fig.~\ref{f:2} and note that possibly $a_i=b_i$. 
The {\bf weight} $w(e)$ of an edge $e \in I$ is the number of cores it covers.
The {\bf weight of a white chain $\SS$} is $w(\SS)=\sum_{C \in \CC} d_{I_\SS}(C)$;
note that $w(e) \le 2$ for any $e \in I$ and thus $w(\SS) \le 2(\ell+1)$.
\end{definition}

\begin{figure} \centering \includegraphics[scale=1.0]{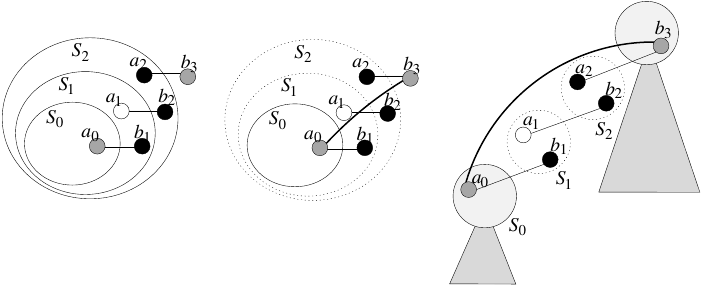}
\caption{Illustration to the shortcut of a white chain of length $\ell=2$.
Here, the black nodes belong to the same core $C$, 
the white node $a_1$ does not belong to any core.
The weight $w(e)$ of the shortcut edge $a_0b_3$ equals to $3$ plus the number of gray nodes that belong to some core.}
\label{f:3} \end{figure}

The laminar family $\LL$ can be represented by a rooted tree $\TT$ with node set $\LL$ and root $V$,
where the parent of $S$ in $\TT$ is the smallest set in $\LL$ that properly contains $S$.
The (unique) edge in $I$ that covers $S$ corresponds to the edge in $\TT$ from $S$ to its parent.
We use for nodes of $\TT$ the the same terminology as for sets in $\LL$;
specifically, nodes of $\TT$ are colored white and black accordingly, 
and a white chain in $\TT$ is a path from a node to its ancestor 
such that all its nodes are white and have degree $2$.

{\bf Short-cutting} a maximal white chain $\SS$ as in Definition~\ref{d:chain} means 
removing from $\LL$ the sets $S_1,\ldots, S_\ell$
and replacing in $I$ the $\ell+1$ edges in $I_\SS$ by the single edge $e=a_0b_{\ell+1}$ of weight $w(e)=w(\SS)$ 
that now has $S_0$ as the witness set; see Fig.~\ref{f:3}.
In the tree representation $\TT$ of $\LL$ this means that we replace the white chain -- the edges in $I_\SS$
and the nodes $S_1,\ldots , S_\ell$ by a new ``shortcut edge'' $e$ of weight $w(e)=w(\SS)$ between the sets that own $a_0$ and $b_\ell$.

Now let us consider the rooted weighted {\bf shortcut tree} $T=(B \cup W,I),w,r$ 
($B$ is the set of {\bf black nodes} and $W$ is the set of {\bf white nodes})
obtained from $\TT$ by short-cutting all maximal white chains. 
Let $L$ be the set of leaves of $T$. 
In what follows, note the following.
\begin{enumerate}
\item 
$w(I)=\sum_{C \in \CC} d_I(C)$, namely, $w(I)$ equals the left-hand side of (\ref{e:2C}).
\item
$|B|=|\CC|$; every core is owned by exactly one set in $\LL$, since $V \in \LL$ and since $\LL$ is laminar.
\item
In $T$, every leaf and every non-root node with exactly one child is black; 
we will call any tree that has this property a {\bf black-white} tree. 
In particular, $T$ has no white chain (a path of white nodes that have exactly one child) and thus $|I| \le 2|\CC|$.
\item
$|I|=|W|+|B|-1 \le 2|B|-1$ and $|W| \le |L| \le |B|$, and if $r$ is black or has at least $2$ children then $|W| \le |B|-1$. 
\end{enumerate}

If the original tree has no white chain of length $> \ell$ then $w(e) \le \ell$ for all $e \in I$,
and thus $\sum_{C \in \CC} d_I(C)=w(I) \le  2(\ell+1) \cdot 2 |\CC|$.
BCGI \cite{BCGI} showed that the maximum possible length of a white chain is $\ell=3$,
which gives the bound $w(I) \le 16|\CC|$. 
To improve this bound the following was proved in \cite{N-prel}.

\begin{lemma}[\cite{N-prel}] \label{l:9}
$w(\SS) \le 5$ for any white chain $\SS$ and if $w(\SS)=5$ then $S_0$ is black.
\end{lemma}

This immediately implies $w(I) \le 10$ (this is the bound that was explicitly stated in \cite{N-prel}), 
but in fact it is also easily implies that $w(I) \le 9|B|$.
To see this, let $t$ be the number of edges of weight $5$. 
Then $t \le |B|$ (since by Lemma~\ref{l:9} the tail of every edge of weight $5$ is black) and since $|W| \le |B|$ we get
\[
w(I) \le 5t+4(|W|+|B|-1-t) \le t+4(2|B|-1) \le 9|B|-4 \ .
\]

In the next section we will describe how to improve the analysis of the approximation ratio 
fro $9$ to $7$.

\section{A 7-approximation for pliable families (Theorem~\ref{t:1})} \label{s:t}

Let $T=(B \cup W,I),w$ be a shortcut tree with root $r$ and leaf set $L$.
For two paths $P,P'$ of $T$ we will write $P \prec P'$ 
if the nodes of $P$ are descendants of the nodes of $P'$.
We will say that an edge of $T$ is {\bf heavy} if it has weight $\ge 3$.
An ordered pair $(e,e')$ of heavy edges is a {\bf bad pair} if $e \prec e'$ 
and there is no black node between $e$ and $e'$. 
Similarly, given two maximal white chains $\SS,\SS'$ we will write $\SS \prec \SS'$ 
if in $\TT$ the nodes of $\SS$ are descendants of the nodes of $\SS'$, 
say that a maximal white chain $\SS$ is heavy if $w(\SS) \ge 3$, 
and say that a pair of heavy maximal white chains $(\SS,\SS')$ is a bad pair if $\SS \prec \SS'$ 
and there is no black set between $S_\ell$ and $S'_0$. 
The following lemma proves the desired bound in the case when there are no bad pairs.

\begin{lemma} \label{l:7}
If $T$ has no bad pair then $w(I) \le 7|B|-2$.
\end{lemma}
\begin{proof}
Let $t$ be the number of heavy edges. 
There are exactly $|I|-t=|W|+|B|-1-t$ non-heavy edges,
hence since $|W| \le |B|$ we have
$$
w(I) \le 5t+2(|W|+|B|-1-t) =3t+2(|W|+|B|-1) \le 3t+2(2|B|-1) \ .
$$
Since all leaves are black and since there is no bad pair, 
we can assign to every heavy edge the closest descendant black node,
and no black node will be assigned twice. Consequently, $t \le |B|$.
Thus we get $w(I) \le 3t+2(2|B|-1) \le 7|B|-2$,
concluding the proof.
\end{proof}

We will prove the following.

\begin{lemma} \label{l:pliable}
Let $(e,e')$ be a bad pair. Then: 
\begin{enumerate}
\item
$w(e)+w(e') \le 7$.
\item
There is no heavy edge on the path between $e$ and $e'$.
\end{enumerate}
\end{lemma}

Note that Lemma~\ref{l:pliable} does not imply that the bad pairs are pairwise disjoint;
if $(e,e')$ is a bad pair then $(e,e')$ is the unique bad pair that contains $e$, 
but there can be many bad pairs $(e_1,e'), \ldots, (e_q,e')$ that contain $e'$. 
Still, Theorem~\ref{t:1} easily follows from Lemmas \ref{l:pliable} and \ref{l:7} 
by a simple manipulation of weights.
For every edge $e'$ that appears as an upper edge in some bad pair, choose one such bad pair 
$(e,e')$ and change the weights of $e'$ to be $2$ and the weight of $e$ to be $w(e)+w(e')-2 \le 5$. 
This operation does not change the maximum weigh nor the total weight, and after it there are no bad pairs,
so  there is a black node between any two ancestor-descendant heavy edges. 
Theorem~\ref{t:1} now follows from Lemma~\ref{l:7}. 
Furthermore, the proof shows that if the bound $w(I) \le 7|B|$ is asymptotically tight, 
then there exists a tight example without bad pairs. 

In the rest of this section we prove Lemma~\ref{l:pliable}. 
Note that in terms of white chains Lemma~\ref{l:pliable} says that 
if $(\SS,\SS')$ is a bad pair of white chains then: 
\begin{enumerate}
\item
$w(\SS)+w(\SS') \le 7$.
\item
There is no heavy maximal white chain between $\SS$ and $\SS'$.
\end{enumerate}

For the proof we will need the following property of white sets.

\begin{lemma} \label{l:Si}
Let $S_{i-1}$ be a child of a white set $S_i \in \LL$ and let $C \in \CC$. 
If $C \cap S_{i-1}$ and $C \sem S_{i-1}$ are both non-empty then $C$ crosses both $S_i,S_{i-1}$.
Furthermore, if $S_{i-1}$ is the unique child of $S_i$ then $S_i \sem S_{i-1} \subset C$.
\end{lemma}
\begin{proof}
Since $S_i$ is white (and thus doesn't own $C$), $C \sem S_i \neq \empt$. 
Thus $C$ crosses both $S_{i-1},S_i$, by Lemma~\ref{l:CS}, 
Now suppose that $S_{i-1}$ is the unique child of $S_i$.
Let $D=S_i \sem (C \cup S_{i-1})$.
By property $(\gamma)$  either $D=\empt$ or $D \in \FF$.
If $D=\empt$ then we are done. Else, $D \in \FF$ and thus $D$ contains a core $C' \in \CC$,
that is owned by a descendant of $S_i$ disjoint to $S_{i-1}$. 
This contradicts that $S_i$ has a unique child. 
\end{proof}

Let $\SS$ be a maximal white chain as in Definition~\ref{d:chain} and let $C \in \CC$. 
The following is proved in \cite{N-prel}; we provide a proof for completeness of exposition.

\begin{lemma} [\cite{N-prel}] \label{l:a}
If $S_0 \cap C \ne \empt$ then either $a_1,b_1 \in C$ or a descendant of $S_0$ or $S_0$ owns $C$;
consequently, if $a_0 \in C$ then $S_0$ owns $C$. For $i \ge 1$ the following holds: 
\begin{itemize}
\item[{\em (i)}]
If $a_i \in C$ then $\ell=i$. 
\item[{\em (ii)}]
If $a_i \notin C$ and $b_i \in C$ then $\ell \in \{i,i+1\}$;
furthermore, if $\ell=i+1$ then $a_{i+1},b_{i+1} \in C$.
\end{itemize}
\end{lemma}
\begin{proof}
If $S_0 \cap C \ne \empt$ and $C$ is not owned by $S_0$ or by a descendant of $S_0$,
then by Lemma~\ref{l:Si}, $\{a_1,b_1\} \subs S_1 \sem S_0 \subset C$.
Now assume that $a_0 \in C$ and suppose to the contrary that $S_0$ does not own $C$. 
Then $C \sem S_0 \neq \empt$. By Lemma~\ref{l:Si}, $S_1 \sem S_0 \subset C$, hence $b_1 \in C$. 
Thus the edge $a_0b_1$ has both ends in $C$, contradicting the assumption the every edge in $I$ covers some $C \in \CC$.

We prove (i). If $S_{i+1}$ exists then by Lemma~\ref{l:Si} $b_{i+1} \in C$,
contradicting the assumption the every edge in $I$ covers some $C \in \CC$.

We prove (ii). If $S_{i+1}$ exists then by Lemma~\ref{l:Si} $a_{i+1}, b_{i+1} \in C$, and $\ell=i+1$ follows from part (i). 
\end{proof}

Let $U=\bigcup_{C \in \CC}C$ be the set of those nodes that belong to some core.
Using Lemma~\ref{l:a}, we obtain the following partial characterization of heavy maximal white chains.

\begin{lemma} \label{l:cases}
If $\SS$ is a heavy maximal white chain then exactly one of the following holds.
\begin{enumerate}
\item
$\ell=1$ and at least $3$ among $a_0,b_2,a_1,b_2$ are in $U$. 
\item
$\ell=2$, $a_1 \notin U$, and one of the following holds:
\begin{enumerate}
\item
$b_1,b_2,a_2 \in C$ for some $C \in \CC$.
\item
$b_1 \notin U$, $a_0,b_2 \in U$, and at least one of $a_2,b_3$ is in $U$.
\end{enumerate}
\item
$\ell=3$, $a_1,b_1 \notin U$, and $b_2,b_3,a_3 \in C$ for some $C \in \CC$.
\end{enumerate}
\end{lemma}
\begin{proof}
The case $\ell=1$ is obvious.
If $\ell=2$ then $a_1 \notin U$, by Lemma~\ref{l:a}. 
If $b_1 \in C$ for some $C \in \CC$ then by Lemma~\ref{l:a} $a_2,b_2 \in C$ and we arrive at case (2a).
Else, $b_1 \notin U$ and since $a_1 \notin U$ we must have $a_0,b_2 \in U$ (since every edge has at least one end in $U$), 
and we arrive at case (2b).

If $a_1,b_1,a_2 \notin U$, then $b_2 \in U$ (since the edge $a_1b_2$ has an end in $U$), 
which by Lemma~\ref{l:a} implies $\ell=3$ and $b_2,b_3,a_3 \in C$ for some $C \in \CC$.
\end{proof}

\begin{figure} \centering \includegraphics[scale=1.1]{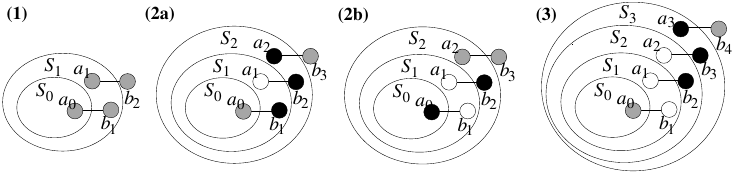}
\caption{The cases in Lemma~\ref{l:cases}.
Black nodes are in $U$, white nodes are not in $U$, while gray nodes may or may not be in $U$.}
\label{f:2} \end{figure}

\begin{figure} \centering \includegraphics[scale=0.95]{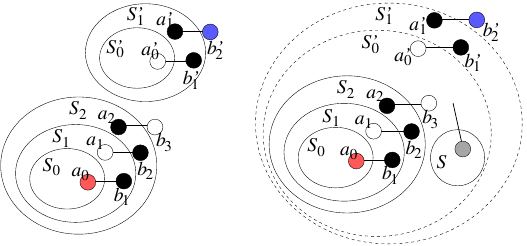}
\caption{Illustration of a bad pair $(\SS,\SS')$ with $w(\SS)+w(\SS')=7$. 
Blue and red nodes belong to distinct cores, while all black nodes belong to the same core. }
\label{f:bad} \end{figure}

\begin{lemma} \label{l:12}
Let $\SS=(S_0, \ldots,S_\ell)$ and $\SS'=(S'_0, \ldots,S_{\ell'})$ be two heavy white chains with edges
$a_0b_1, \ldots, a_\ell b_{\ell+1}$ and  $a'_0b'_1, \ldots, a'_{\ell'} b'_{\ell'+1}$, respectively. 
If $\SS \prec \SS'$ and there is no black set between $\SS$ and $\SS'$ then (see Fig.~\ref{f:bad}):  
\begin{enumerate}
\item
$w(\SS')=3$, $\ell'=1$, and $a'_0 \notin U$.
\item
$w(\SS) \le 4$ and if $\ell =1$ then $a_0 \in U$.
\end{enumerate}
\end{lemma}
\begin{proof}
Consider the lower chain $\SS$. By Lemma~\ref{l:cases}, 
one of the nodes $a_1,b_1, \ldots,a_\ell,b_\ell$ is in $C$ for some $C \in \CC$.
The core $C$ is not owned by sets in $\SS \cup \SS'$ nor by sets between $\SS$ and $\SS'$, since all these sets are white.
Thus $C$ crosses all sets in the upper chain $\SS'$, and in particular the set $S'_1$. 
By Lemma~\ref{l:Si} $S'_1 \sem S'_0 \subs C$, 
and in particular $a'_1 \in C$. This implies $\ell=1$, by Lemma~\ref{l:a}. 
Moreover, $a'_0 \notin U$, as otherwise by Lemma~\ref{l:a} $S'_0$ is black,
contradicting the assumption that there is no black set between $\SS$ and $\SS'$. 
This proves part~1. 

For part~2, we claim that one of $a_\ell,b_{\ell+1}$ is not in $U$. 
Suppose to the contrary that $a_{\ell} \in C$ and $b_{\ell+1} \in C'$ for some distinct $C,C' \in \CC$. 
Then each of $C,C'$ crosses all the sets in $\SS'$, and in particular the first set $S'_1$. 
Thus by Lemma~\ref{l:Si} $S'_1 \sem S'_0 \subs C \cap C'$, 
and in particular $a'_1,b'_1 \in C \cap C'$. This contradicts that the cores are disjoint.
Consequently, one of $a_\ell,b_{\ell+1}$ is not in $U$, which implies $w(\SS) \le 4$, by Lemma~\ref{l:cases};
note that $w(\SS)=5$ is possible only in case (3) of Lemma~\ref{l:cases} when $a_3,b_4 \in U$. 
If $\ell=1$, then $a_0 \in U$ as otherwise $\SS$ is not heavy.
\end{proof}

Lemma~\ref{l:12} already implies the first part~1 of Lemma~\ref{l:pliable}, that $w(\SS)+w(\SS') \le 7$.
We will show that it also implies part~2. 
Suppose to the contrary that there is another maximal white chain $\SS''$ between $\SS$ and $\SS'$.
To obtain a contradiction we apply Lemma~\ref{l:12} twice, as follows.
\begin{itemize}
\item
Since $\SS \prec \SS''$, Lemma~\ref{l:12} implies $\ell''=1$ and  $a''_0 \notin U$.
\item
Since $\SS'' \prec \SS'$, Lemma~\ref{l:12} implies that if $\ell''=1$ then $a''_0 \in U$.
\end{itemize}
In the first application $a''_0 \notin U$ while in the second $a''_0 \in U$, arriving at a contradiction.

\medskip

This concludes the proof of Lemma~\ref{l:pliable}, 
and thus also of Lemma~\ref{l:2C} and Theorem~\ref{t:1}.

\begin{figure} \centering \includegraphics[scale=0.90]{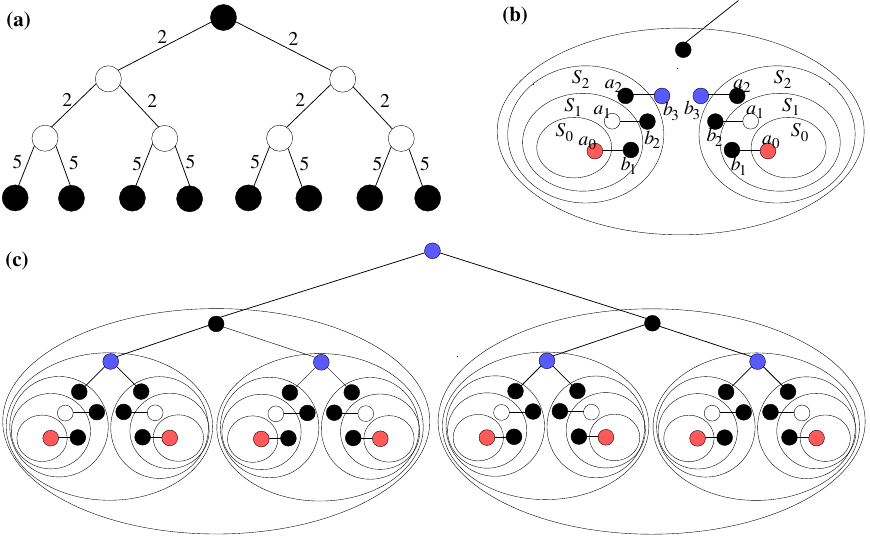}
\caption{Construction of a tree $\TT$ of weight $7|L|-2$ and a set of $|L|+2$ cores.
(a) The shortcut tree. (b) The gadgets. (c) The laminar family and the cores.}
\label{f:tight7} \end{figure}

The following example shows that the bound in (\ref{e:7}) is asymptotically tight.
The shortcut-tree is a binary tree with black nodes $B=L \cup \{r\}$ and weights $5$ for leaf edges 
while all the other edges have weight $2$; see Fig.~\ref{f:tight7} for an illustration for the case $|L|=8$.
To materialize this tree in terms of the laminar family and cores, do the following.
\begin{itemize}
\item
Replace every leaf edge by the gadget as in case (2a) in Lemma~\ref{l:cases} 
where $a_0,b_3 \in U$ belong to distinct cores. 
\item
Every other edge will connect two distinct cores, when the same cores are used for distinct edges.
\end{itemize}
Every red node is a core (these cores are distinct), and there are two additional cores -- 
one contains all black nodes and the other all blue nodes; these two cores are owned by the root $V$.
The number of cores is $|L|+2$, while the total weight is $5|L|+2(|L|-1)=7|L|-2$. 

\newpage

\section{Improved approximation for sparse families (Theorem~\ref{t:2})} \label{s:sparse}

Recall that a set family $\FF$ is sparse if for any edge set $J$, 
every set $S \in \FF^J$ crosses at most one $\FF^J$-core.
This implies that if $\FF$ is sparse, then so is any residual family $\FF'$ of $\FF$.
Due to this and Lemmas \ref{l:res} and \ref{l:2C'},
to prove that the WGMV algorithm achieves approximation ratio $6$ for a $\ga$-pliable sparse family $\FF$, 
it is sufficient to prove the following. 

\begin{lemma} \label{l:6C}
Let $I$ be an inclusion minimal cover of a $\gamma$-pliable sparse set family $\FF$ such that 
every edge in $I$ covers some $C \in \CC$. Then 
\begin{equation} \label{e:6}
\sum_{C \in \CC} d_I(C) \leq 6 |\CC|-2 \ .
\end{equation}
\end{lemma}

In the proof of Lemma~\ref{l:6C} we will use the first part of the following lemma.

\begin{lemma} \label{l:sparse}
If $\FF$ is sparse then for any edge $e$ of the shortcut tree the following holds: 
\begin{itemize}
\item
If $w(e)=5$ then both ends of $e$ are black.
\item
If $w(e)=4$ then at least one end of $e$ is black.
\end{itemize}
\end{lemma}
\begin{proof}
Let $\SS$ be a white chain.
By Lemma~\ref{l:cases}, if $w(\SS)=5$ then $a_0,a_\ell,b_{\ell+1} \in U$; 
see cases (2a) and (3) in Figure~\ref{f:2} and Lemma~\ref{l:cases}. 
Thus by Lemma~\ref{l:a} $S_0$ is black (since $a_0 \in U$).
Note that $a_\ell,b_{\ell+1}$ belong to distinct cores, say $a_\ell \in C$ and $b_{\ell+1} \in C'$.
Let $S$ be the parent of $S_\ell$. 
Since $\FF$ is sparse, at least one of $C,C'$ cannot cross $S$, and thus is owned by $S$.
Hence $S$ is also black.
If $w(\SS)=4$ then 
$a_0 \in U$ and then $S_0$ is black, or 
$a_\ell,b_{\ell+1} \in U$ and then the parent $S$ of $S_\ell$ is black.
\end{proof}

Let $T=(W \cup B,I),r,w$ be the shortcut tree; 
recall that $T$ is a black-white tree, namely, every non-root node with exactly one child is black.
We already know that $w(e) \le 5$ and if $w(e)=5$ then $e$ has its lower end in $B$ (by Lemma~\ref{l:9});
Lemma~\ref{l:sparse} adds the property that if $w(e)=5$ then $e$ has both ends in $B$.
Furthermore, Lemma~\ref{l:12} implies the following.

\begin{corollary} \label{c:bad}
For any bad pair $(e,e')$ of $T$ the following holds.
\begin{enumerate}
\item
There is no heavy edge between $e$ and $e'$.
\item
$w(e) \le 4$, $w(e')=3$ and $e'$ has an upper end in $B$.
\end{enumerate}
\end{corollary}

For a heavy edge $e$ let $B_e$ be the set of black nodes $b \in B$ such that
$b$ is a descendant of $e$ and there is no heavy edge between $e$ and $b$.
Note the following.
\begin{itemize}
\item
$B_e=\empt$ if and only if $e$ is an upper edge of a bad pair. 
\item
$B_e \cap B_f=\empt$ for distinct $e,f$.
\end{itemize}

Let $H^*$ be the set of heavy edges $e$ such that $e$ is not the upper edge of a bad pair,
so $B_e \ne \empt$ for all $e \in H^*$. 
For every $e \in H^*$ choose one node $b_e \in B_e$ and let 
$B^*=\{b_e:e \in H^*\}$ be the set of chosen black nodes.  
Now we will prove the following lemma, that implies Lemma~\ref{l:6C}.

\begin{lemma} \label{l:tree''}
$w(I) \le 3|B|+|L| +2|B^*|-2$ and $w(I) \le 3|B|+|L| +2|B^*|-4$ 
if $r$ is black or has at least $2$ children.
\end{lemma}
\begin{proof}
Assign tokens to nodes in $B$ as follows:
\begin{itemize}
\item
$4$ tokens to every node in $L$.
\item
$3$ tokens to every node in $B \sem L$.
\item
$2$ additional tokens to every node in $B^*$.
\end{itemize}
The number of tokens is 
$4|L|+3|B \sem L|+2|B^*|=3|B|+|L|+2|B^*|$.
We will show that these tokens can be redistributed such that 
every $e \in I$ gets $w(e)$ tokens and some tokens remain.
 
For each $e \in H^*$ use $2$ tokens of $b_e$ to reduce the weight of $e$ by $2$. 
Then we have the following.
\begin{itemize}
\item
Every leaf has $4$ tokens and every internal black node has $3$ tokens. 
\item
The maximum weight is $3$ since initially the maximum weight  was $5$ and 
since the upper edge of every bad pair has weight exactly $3$, by Corollary~\ref{c:bad}.
\item
Every edge of weight $3$ has its upper end in $B$,
since every edge of weight $5$ and the upper edge of every bad pair 
has its upper end in $B$, by Lemma~\ref{l:sparse} and Corollary~\ref{c:bad}. 
\end{itemize}

For $v \in V$ let $T_v$ be the rooted subtree of $T$ that consists of $v$ and its descendants. 
We claim that for any $v \ne r$, the tokens of $T_v$ can be redistributed such that 
every edge gets at least $w(e)$ tokens and 
the root $v$ gets $4$ tokens.
The proof is by induction on he height of the tree.
The induction base case, when $v$ is a leaf, is trivial.
Suppose that $v$ is not a leaf and has $p \ge 1$ children.
By the induction hypothesis each child of $v$ has $4$ tokens. 

Suppose that $v$ is white. Then $p \ge 2$. 
Each child of $v$ is connected to $v$ by an edge of weight $\le 2$ 
(since the upper end of every edge of weight $3$ is black),
and can pay for its parent edge and give $2$ tokens to $v$. 
Thus $v$ gets $2p \ge 4$ tokens.

Suppose that $v$ is black, so $v \in B \sem L$. Then $v$ already has $3$ tokens.
Each child of $v$ is connected to $v$ by an edge of weight $\le 3$. 
Thus in this case each child can pay for his parent edge and give $1$ token to $v$.
Thus $v$ gets $3+p \ge 4$ tokens. 

Now let us consider the root $r$ of $T$ that has $p \ge 1$ children.
If $r$ is white then it gets $2p \ge 2$ tokens. 
If $r$ is black then it gets $3+p \ge 4$ tokens, and 
if $r$ has at least $2$ children then it gets $3+p \ge 4$ tokens.
\end{proof}

\begin{figure} \centering \includegraphics[scale=0.89]{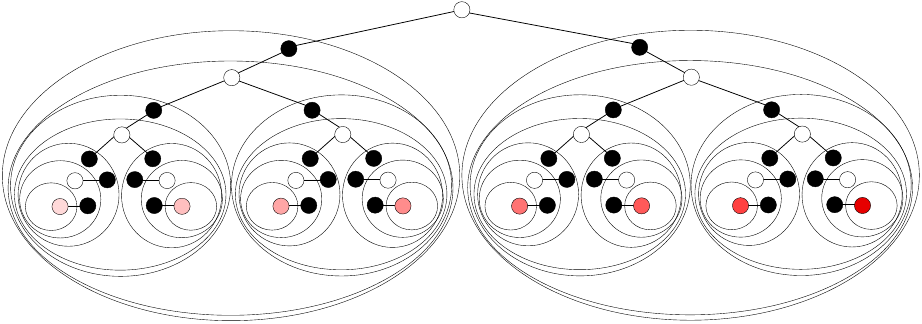}
\caption{Construction of a tree $\TT$ of weight $6|L|-2$ and a set of $|L|+1$ cores.
Any two red nodes belong to distinct cores, while all black nodes belong to the same core.
(a) The shortcut tree. (b) The gadgets. (c) The laminar family and the cores.}
\label{f:tight6} \end{figure}

Lemma~\ref{l:tree''} implies that $w(I) \le 6|B|-2 \le 6|\CC|-2$.
The following example shows that the bound $w(I) \le 6|B|$ is asymptotically tight even when there are no bad pairs.
The shortcut-tree is a binary tree with black nodes $B=L \cup \{r\}$ and weights $4$ for leaf edges 
while all the other edges have weight $2$; see Fig.~\ref{f:tight6} for an illustration for the case $|L|=8$.
To materialize this tree in terms of the laminar family and cores, do the following.
\begin{itemize}
\item
Replace every leaf edge by the Lemma~\ref{l:cases}(2a) gadget, 
where $a_0,\in U$ and $b_3 \notin U$. 
\item
Replace every other edge by the gadget as in case (1) in Lemma~\ref{l:cases} where $a_1,b_1 \in U$ and $a_0,b_3 \notin U$;
this is a ``redundant'' (non-heavy) white chain of weight $2$.    
\end{itemize}
Every node colored by a shade of red is a core 
(we used distinct shades of red to indicate that these cores are distinct), 
and there is one additional core -- 
the one that contains all black nodes; this core is owned by the root $V$.
The number of cores is $|L|+1$, while the total weight is $4|L|+2(|L|-1)=6|L|-2$. 

Note that in this example every member of the laminar family is crossed by at most one core, 
and that there are no bad pairs in this example.

\medskip

To prove the second part of Theorem~\ref{t:2} we prove the following.

\begin{lemma} \label{l:5C}
Let $I$ be an inclusion minimal cover of a $\gamma$-pliable sparse $\be$-crossing set family $\FF$ such that 
every edge in $I$ covers some $C \in \CC$. Then 
\begin{equation} \label{e:5}
\sum_{C \in \CC} d_I(C) \le 5|\CC| +\f{|\CC|}{1+1/\be} =\left(6-\f{1}{\be+1}\right) |\CC| \ .
\end{equation}
\end{lemma}
\begin{proof}
We claim that  
\[
|\CC| \ge |L|+|L \cap B^*|/\be \ge |L \cap B^*| \cdot (1+1/\be)
\]
To see this consider the set of edges 
$I^*=\{e \in H^*:b_e \in L \cap B^*\}$; namely, 
$e \in I^*$ if $e \in H^*$ (so $e$ is heavy but is not the upper edge of a bad pair) and 
the black node $b_e$ assigned to $e$ is a leaf.
Note the following.
\begin{itemize}
\item
For any two edges in $I^*$, none of them is a descendant of the other, 
since if $e \in I^*$ has a descendant edge $f \in I^*$,
then $b_e$ is between $e$ and $f$ contradicting that $b_e$ is a leaf. 
\item
For every $e \in I^*$ there is a core $C_e \in \CC$ that crosses all the sets in the white chain of $e$.
\end{itemize}
For every $e \in I^*$ chose some set $S_e$ from the white chain of $e$. 
The sets $S_e$ are pairwise disjoint. Since $\FF$ is $\be$-crossing, 
in any set of $\be+1$ edges from $I^*$ there are two edges $e,f$ with $C_e \neq C_f$.
Consequently, $|\{C_e:e \in I^*\}| \ge |I^*|/\be=|L \cap B^*|/\be$. 
There are also $|L|$ cores contained in leaves of $T$.
Thus 
$|\CC| \ge |L|+  |\{C_e:e \in I^*\}| \ge |L|+|L \cap B^*|/\be$.

\medskip

By Lemma~\ref{l:tree''} $w(I) -3|B|-|B^*| \le |L| +|B^*|$. Thus since $L \cup B^* \subs B$ we get
\[
w(I) -3|B|-|B^*| \le |L| +|B^*| = |L \cup B^*|+|L \cap B^*| \le |B|+|\CC|/(1+1/\be)
\]
Consequently, 
$w(I) \le 4|B|+|B^*|+|\CC|/(1+1/\be) \le 5|\CC| +|\CC|/(1+1/\be)$,
as required.
\end{proof}

Note that we proved approximation ratio $\f{6\be+5}{\be+1}$.
We can provide an example without edges of weight $5$ and without bad pairs 
such that $w(I) \approx \f{6 \be}{\be+1}$. 
Consider the example in Fig.~\ref{f:tight6}.
Suppose that $|L|=2^i$ and $\be=2^j$ for some $0 \le j<i$. 
Fig.~\ref{f:tight-be} illustrates the construction for $j=1$ and $i=3$, where 
two nodes are colored by the same color if they belong to the same core. 
Consider the minimal rooted subtrees of $T$ with exactly $\be=2^j$ leaves. 
Each such subtree will be exactly as in the example in Fig.~\ref{f:tight6} --
the leaves are colored by distinct colors, while the other nodes in $U$ all have the same color,
which we call ``the color of the subtree''.
For each subtree, its root is not in $U$ (the union of all cores), 
while the parent of the root is in $U$ and has the color of the subtree.
The grandparent of the root is again a node not in $U$ and joins two subtrees;
the parent of the grandparent will is colored by the color of one of these subtrees. 
In a similar manner we propagate the colors upwards, where 
a node not in $U$ joins two subtrees and its parent is colored by the color of one of these two subtrees. 
Note that such ``coloring'' does not violate the $\be$-crossing property. 
In the constructed example, $w(I)=6|L|-2$ (the weight is the same ias in the example in Fig.~\ref{f:tight6}) 
and $|\CC|=|L|+|L|/\be$, hence when $|L|$ is large $w(I)/|\CC| \approx \f{6}{1+1/\be}=\f{6 \be}{\be+1}$.

\begin{figure} \centering \includegraphics[scale=0.85]{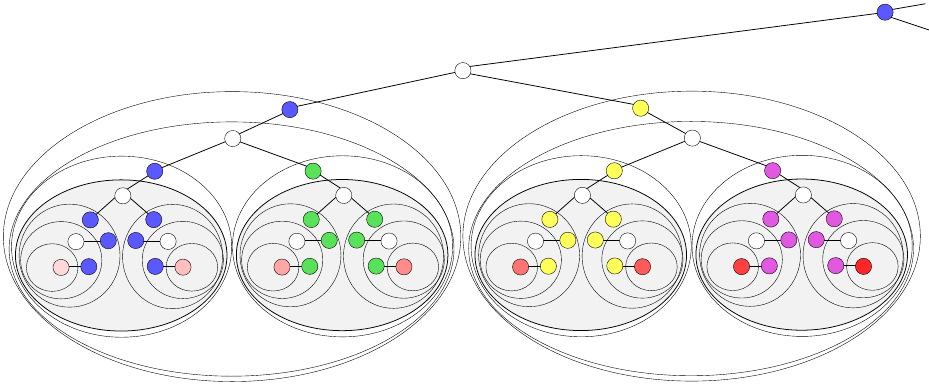}
\caption{Illustration of the construction of a tree with $w(I)=6|L|-2$ and $|\CC|=|L|+|L|/\be$ with $|L|=2^3$ leaves and $\be=2^1$.}
\label{f:tight-be} \end{figure}


\end{document}